\documentclass[aps,pra,showpacs,nofootinbib,twocolumn,nofootinbib,superscriptaddress]{revtex4-1}
\usepackage{amsmath,amsthm,amssymb}
\usepackage{mathrsfs}
\usepackage{color}
\usepackage{dcolumn}
\usepackage{multirow}
\usepackage{array}
\usepackage{dsfont}
\usepackage{graphicx}
\usepackage{calc}
\makeatletter
\def\be{\begin{equation}}
\def\ee{\end{equation}}
\def\bea{\begin{eqnarray}}
\def\eea{\end{eqnarray}}
\def\ben{\begin{equation*}}
\def\een{\end{equation*}}
\def\bean{\begin{eqnarray*}}
\def\eean{\end{eqnarray*}}
\def\bma{\begin{mathletters}}
\def\ema{\end{mathletters}}
\def\bi{\begin{itemize}}
\def\ei{\end{itemize}}

\newtheorem{theorem}{Theorem}

\newtheorem{conjecture}[theorem]{Conjecture}

\newtheorem{definition}[theorem]{Definition}

\newtheorem{observation}[theorem]{Observation}

\begin{document}
\title{Witnessing Genuine Mutipartite Non-locality}

\author{Some Sankar Bhattacharya}
\email{somesankar@gmail.com}\affiliation{Physics \& Applied Mathematics Unit, Indian Statistical Institute, 203 B.T. Road, Kolkata- 700108, India}
\author{Arup Roy}
\email{arup145.roy@gmail.com}\affiliation{Physics \& Applied Mathematics Unit, Indian Statistical Institute, 203 B.T. Road, Kolkata- 700108, India}
\author{Amit Mukherjee}
\email{amitisiphys@gmail.com}\affiliation{Physics \& Applied Mathematics Unit, Indian Statistical Institute, 203 B.T. Road, Kolkata- 700108, India}
\author{Ramij Rahaman}
\email{ramijrahaman@gmail.com}\affiliation{Department of Mathematics, University of Allahabad, Allahabad 211002, U.P., India}

\begin{abstract}
Genuine multipartite nonlocality is a salient feature of quantum systems, empowering the security of multi-party device independent cryptographic protocols. Given a correlation, characterizing and detecting \emph{genuineness} have been subjected to recent studies. 
In this regard, we propose a Hardy-type argument which is able to detect genuine $n$-way nonlocality of arbitrary quantum systems. To understand the strength of this argument we also study the optimal success probability of the argument in a minimally constrained theory, namely the {\em generalized no-signaling theory}.
\end{abstract}
\maketitle
\section{Introduction}
The outcome of local measurements on spatially separated quantum systems can not be classically simulated using shared randomness. This can be demonstrated by taking one of the two possible paths: by collecting measurement statistics or by performing one shot measurements. Bell's inequality(BI) \cite{Bell64} is a tool of this first kind of path to expose the non classical feature of quantum mechanics. Whereas a direct contradiction between Quantum mechanics and local realism was found using GHZ state by \emph{All vs Nothing} type of proof \cite{Greenberger1990}. Though their approach is very elegant and simple but still it relates eight dimensional Hilbert space. In 1992 L. Hardy \cite{Hardy1993,Hardy1992} provided a no-go theorem for local-realistic hidden variable(LHV) models for almost all pure two qubit states except maximally entangled state. This argument does not use statistical inequalities involving expectation values.
\par
Although quantum mechanics(QM) is inconsistent with local realistic theory but still it is impossible to exploit QM for faster than light communication. Interestingly, quantum theory is not the single candidate which violates Bell Inequality and respects the constraint of relativistic causality. These theories are known as `Generalized No-signaling Theories(GNST)'. There exists GNST which is more non-local than QM but nevertheless respects relativistic causality\cite{Popescu1994}. This kind of post-quantum theory may help us to better understand the `restricted non-locality' of QM. We investigate multi-partite extensions of Hardy's non-locality argument in GNST.
\par
The structure of multi-partite entanglement is not as simple as the bipartite case. Many features of the multi-partite state is not clearly understood still now in sharp contrast to bipartite scenario\cite{Bancal2011,Bancal2013}. A n-partite entangled state will be called genuine iff the state is not m-separable with respect to any m-partition($m\le n$) of the subsystems. Being useful resource for computation\cite{Raussendorf2001}, simulation\cite{Lloyd1996}, metrology\cite{Pezze2009}, study of multi-partite entanglement is a field of latest attraction. It is even useful for dinning cryptography problem\cite{Rahaman2014}.
\par
Analogous to entanglement theory, multi-partite non-locality is also not very easy to understand compared to the bipartite cases. In recent years there have been a number of studies relating the genuineness of entanglement and non-locality. To show whether this two sets are identical or not is in sharp contrast with the bipartite scenario. In \cite{Rahaman2014} the authors provide modified Hardy type argument for arbitrary n-partite system and conclude that only genuine multi-partite entangled state satisfy that modified Hardy argument. In \cite{Chen2014} Chen \emph{et.al.} have also shown that Hardy type argument can be used to wit the genuineness of multi-partite non-locality without inequality. They also prove that all pure entangled symmetric N-qubit states $(N \ge 2)$are genuinely multi-partite nonlocal. Further the authors conjecture that all pure genuine entangled states are genuine multi-partite nonlocal. But their argument works as a witness for n- qubit scenario. We have extended this argument for n- qudit scenario. We also investigate how much success probability it allows in No signaling theory. Whether it changes with dimension in contrast with other multipartite Hardy type argument.
\par
In this article we have extended the question of witnessing a multipartite correlation is genuine or not to an arbitrary higher dimensional scenario. It also provides an intuitive understanding of their violations in minimally restricted theory i.e in no signaling theory. In this paper first we define the different notions of genuine multipartite non locality followed by detecting genuine $N$-way non locality without inequality for qubit scenario. Then we provide the witness argument for arbitrary dimension i.e for qudits. Finally we study the nature of this argument in a minimally constrained theory.

\section{Genuine Multipartite Nonlocality}
The concept of the {\it genuineness} of multipartite nonlocality as a unique feature of multi-party systems was first proposed by Svetlichny \cite{Svetlichny1987, Seevinck2002} in the following form
\begin{definition}(Svetlichny)\label{s2}
If the joint probability distribution $P(abc|XYZ)$ can be written in the hybrid local-nonlocal form
\begin{align}\label{svetlichnylocal}
&P(abc|XYZ) = \sum_\lambda q_{\lambda}\,Q_\lambda(ab|XY)\,R_\lambda(c|Z) + \nonumber\\
& \sum_\mu q_{\mu}\,Q_\mu(ac|XZ)\,R_\mu(b|Y) + \sum_\nu q_{\nu}\,Q_\nu(bc|YZ)\,R_\nu(a|X),
\end{align}
where $0\leq q_{\lambda},q_{\mu},q_{\nu} \leq 1$ and $\sum_\lambda q_{\lambda}+\sum_\mu q_{\mu}+\sum_\nu q_{\nu}=1$. Then the correlations are $S_2$-local. Otherwise, they are genuinely 3-way Svetlichny nonlocal.
\end{definition}
Let us consider an underlying hidden state $\lambda$ such that the probabilities arising from $\lambda$ depend only on the temporal \emph{ordering} determined with respect to a fixed reference frame and not on the exact timing of measurements. If Alice's measurement precedes Bob's, the probabilities are given by
\begin{equation}\label{bipartitealicefirst}
P_{\lambda}^{A<B}(ab|XY).
\end{equation}
On the other hand, when Bob's measurement precedes Alice's, the correlations may be different with probabilities given by
\begin{equation}\label{bipartitebobfirst}
P_{\lambda}^{B<A}(ab|XY).
\end{equation}
However if one considers signaling hidden states, there exist arguments leading to paradoxes, making use of the knowledge of measurement outcomes of one party to determine measurement choices on others. These paradoxes are avoided if the correlations $P_{\lambda}^{A<B}(ab|XY)$ and $P_{\lambda}^{B<A}(ab|XY)$ are at most 1-way signaling, with $P_{\lambda}^{A<B}(ab|XY)$ and $P_{\lambda}^{B<A}(ab|XY)$ satisfying
\begin{align}
P_{\lambda}(a|XY) &= P_{\lambda}(a|XY') \quad \forall a,X,Y,Y' \label{nosigbobtoalice}\\
P_{\lambda}(b|XY) &= P_{\lambda}(b|X'Y) \quad \forall b,Y,X,X'.\label{nosigalicetobob}
\end{align}
respectively. Another way out is to consider only non-signaling hidden states. This suggests the following definition of genuine tripartite nonlocality
\begin{definition}(Bancal et al.\cite{Bancal2013})\label{nonsiggennonloc}
If the joint probability distribution $P(abc|XYZ)$ can be written in the form
\begin{align}\label{genuinenonsignonloc}
&P(abc|xyz) = \sum_\lambda q_\lambda Q_{\lambda}(ab|xy)R_{\lambda}(c|z)+\nonumber \\
&\sum_\mu q_\mu Q_{\mu}(ac|xz)R_{\mu}(b|y)+\sum_\nu q_\nu Q_{\nu}(bc|yz)R_{\nu}(a|x),
\end{align}
where all possible bi-partitions are non-signaling, satisfying both the conditions in(\ref{nosigbobtoalice}) and (\ref{nosigalicetobob}). Then the correlations are $NS_2$-local. Otherwise, they are genuinely 3-way NS nonlocal.
\end{definition}
We use def.\ref{nonsiggennonloc} in the following sections. Here we investigate the question whether a given correlation, which is shared between more than 2 parties is genuine $NS_2$ nonlocal or not. We have extended the scheme which is introduced in \cite{Chen2014} for the arbitrary dimensional scenario. The spirit of the scheme presented by \emph {Chen} et.al \cite{Chen2014} is a stronger form of Hardy \cite{Hardy1992} like proof of nonlocality i.e nonlocality without inequality. But their argument is restricted only in the $2$ dimensional scenario. We extend the `genuine-nonlocality' argument for higher dimension.

\section{Genuine N-way Nonlocality without Inequality}
We are going to exhibit a generalized Hardy type argument for arbitrary $N$ number of parties. This argument allows one to discriminate genuine $N$-party nonlocality from any hybrid model {\em m versus (N-m)} irrespective of dimension. We consider any possible partition in {\em two} subsets of $m$ and $N-m$ parties respectively, with $1\leq m\leq N-1$. Partitions in more number of smaller subsets appear as special cases of these bi-partitions. Here the concept of genuine 3-way nonlocality generalizes to $N$-parties in the following way
\begin{definition}\label{nonsiggennonlocN}
Suppose that the joint probability distribution $P(\Lambda_N|\Omega_N)$ can be written in the form
\begin{equation}\label{genuinenonsignonloc}
P(\Lambda_N|\Omega_N) = \sum_{\{\alpha:|\alpha|=m\}}\sum_\lambda q_{\alpha,\lambda} Q_{\alpha,\lambda}(\Lambda_\alpha|\Omega_\alpha)R_{\bar{\alpha},\lambda}(\Lambda_{\bar{\alpha}}|\Omega_{\bar{\alpha}})
\end{equation}
where each of the bipartition $\{\alpha,\bar{\alpha}\}$ are non-signaling. Then the correlations are $NS_2$-local. Otherwise, we say that they are genuinely N-way NS nonlocal.
\end{definition}

 \subsection{Multipartite Hardy (MH) paradox: Qubit Scenario}
Recently, Chen {\em et al \cite{Chen2014}} presented a relaxed Hardy-type (RMH) test in comparison to the Ref. \cite{Rahaman2014} for detecting the genuine non-locality of $N$ two-level systems as:
\begin{eqnarray}\label{HardyRG}
& P(1\dots 1|\hat{u}_1\dots\hat{u}_N) = q >0,&\nonumber\\
& \forall r ~P(1\dots 111\dots 1|\hat{u}_1\dots \hat{u}_{r-1}\hat{v}_r\hat{u}_{r+1}\dots \hat{u}_N) = 0,&\\
&\forall i\neq j~P(1\dots 2\dots 2\dots 1|\hat{u}_1\dots \hat{v}_i\dots\hat{v}_j\dots\hat{u}_N) = 0.&\nonumber
\end{eqnarray}
Here $j\in \{1,\dots,N\}$ is kept fixed. The second and third equations contain $n$ and $n-1$ conditions respectively. They have shown that the above $2n$ conditions can not be simultaneously satisfied by a correlation which is not genuine nonlocal. Thus only a genuine n-qubit non-local correlation can satisfy (\ref{HardyRG}) and hence pass the test of genuine non-locality. For bipartite scenario the argument boils down to the original Hardy paradox introduced by L. Hardy in \cite{Hardy1992,Hardy1993}. For more than two parties if we replace the third condition of the argument properly then conventional multipartite Hardy paradox appears\cite{Chudhary2010} which can be used for demonstrating multipartite non locality. Chen's argument can be regarded as an natural generalization of Hardy's paradox to detect {\it genuine} multipartite non locality.
\begin{theorem}(Chen et al.)
All n-qubit quantum states that satisfy (\ref{HardyRG}) are genuine non-local.
\end{theorem}
The authors establish this fact via reductio ad absurdum. At first they have taken a correlation which is not genuine multipartite nonlocal and then they arrive at a contradiction with this presumption. This theorem demonstrates that one can detect \emph{genuineness} of nonlocal correlation even without inequality. But one of the limitations of this argument is that it does not answer the question about detecting \emph{genuine} nonlocality for systems with arbitrary dimension. In the next section we provide argument for witnessing genuineness by retaining the spirit of the logical argument in \cite{Chen2014} but for arbitrary dimension.
\subsection{Modified Hardy type argument for arbitrary dimension}
Although a considerable effort has been made in characterizing the {\it genuineness} of multipartite non locality in recent times, the question of witnessing genuineness for arbitrary higher dimensional systems is still a point of interest. Here we propose a higher dimensional extension of $N$-partite relaxed Hardy-type logical argument mentioned in the last sub-section. Consider $N$ subsystems shared among $N$ separated parties and the ith party can measure one of the two observables, $\hat{u}_i$ and $\hat{v}_i$, on the local subsystem. The possible outcomes $x_i$ of each such local (i-th party) measurement can be $1,\dots d_i$. Then the relaxed Hardy-type logical argument can start from the following set of joint probability conditions: 
\begin{align}\label{relH}
& P(1\dots1|\hat{u}_1\dots \hat{u}_N) = q >0,&\nonumber\\
& P(1\dots 1\neg d_r1\dots 1|\hat{u}_1\dots \hat{u}_{r-1}\hat{v}_r\hat{u}_{r+1}\dots \hat{u}_N) = 0,~\forall r &\nonumber\\
& P(1\dots 1d_i1\dots 1d_j1\dots1|
\hat{u}_1\dots\hat{u}_{i-1}\hat{v}_i\hat{u}_{i+1}\dots\nonumber\\&\hat{u}_{j-1}\hat{v}_j\hat{u}_{j+1}\dots\hat{u}_N)= 0~\forall i\neq j
\end{align} 

Here $\neg d_r$ denotes other than $d_r$ it can be any positive integer less or equal to $d_r-1$ and $j\in \{1,....,N\}$ is kept fixed. Now the question is whether (\ref{relH}) can be used to witness the genuineness of multi-partite non-locality in QM. To answer this question one needs to show that all quantum states satisfying the above set of conditions can not be written in a \emph{m versus (N-m)} local form as in (\ref{nonsiggennonlocN}) and we answer in affirmative:

\begin{theorem}
All n-qudit quantum states that satisfy (\ref{relH}) are genuine non-local.
\end{theorem}
\begin{proof}
Consider state $\rho$ satisfying conditions (\ref{relH}), which is not genuinely $N$-way non-local. Thus it is a convex combination of (at least) bi-local states. Each of them is bi-local with respect to some cut, say, $(1,2,...,m)$ vs. $(m+1,m+2,...,N)$. As all Hardy conditions are expressed in terms of probabilities, there must be at least one term in the convex combination which gives a non-zero contribution to the first condition, with $q'>0$. Assume that such a term has the following form $Q(x_1\dots x_m|\hat{x}_1\dots \hat{x}_m)R(x_{m+1}\dots x_N|\hat{x}_{m+1}\dots \hat{x}_N)$. All Hardy conditions must hold for this term. One must have: by the first one
\begin{widetext}
\bean
&Q_{\lambda}(1\dots 1|\hat{u}_1\dots \hat{u}_m)R_{\lambda}(1\dots 1|\hat{u}_{m+1}\dots \hat{u}_N)=q'>0&\\
&\implies ~Q_{\lambda}(1\dots 1|\hat{u}_1\dots \hat{u}_m)=q'_1>0\mbox{ and } R_{\lambda}(1\dots 1|\hat{u}_{m+1}\dots \hat{u}_N)=q'_2>0 \mbox{ [$q'=q'_1q'_2$, say]},&\label{relation1}
\eean
then the middle condition gives us
\ben\begin{split} \forall r\leq m, ~ Q_{\lambda}(1\dots1\neg d_r1\dots 1|\hat{u}_1\dots \hat{u}_{r-1}\hat{v}_r\hat{u}_{r+1}\dots\hat{u}_m)=0 &\implies Q_{\lambda}(1\dots 1d_r1\dots 1|\hat{u}_1\dots\hat{u}_{r-1}\hat{v}_r\hat{u}_{r+1}\dots\hat{u}_m)>0 \\ \forall ~i>m,~R_{\lambda}(1\dots 1\neg d_{i}1\dots 1|\hat{u}_{m+1}\dots \hat{u}_{i-1}\hat{v}_{i}\hat{u}_{i+1}\dots\hat{u}_{N}) =0&\implies R_{\lambda}(1\dots 1 d_{i} 1\dots 1|\hat{u}_{m+1}\dots \hat{u}_{i-1}\hat{v}_{i}\hat{u}_{i+1}\dots \hat{u}_{N}) >0.\end{split}\een
Therefore,
\be\label{cond1}
Q_{\lambda}(1\dots 1d_r1\dots 1|\hat{u}_1\dots \hat{u}_{r-1}\hat{v}_r\hat{u}_{r+1}\dots\hat{u}_m)
R_{\lambda}(1\dots 1d_{i}1\dots 1|\hat{u}_{m+1}\dots\hat{u}_{i-1}\hat{v}_{i}\hat{u}_{i+1}\dots \hat{u}_{N})>0,~
\forall r\leq m,~\&~\forall i>m.
\ee
Now, for $j\leq m$ the last condition of (\ref{relH}) gives us
\ben\label{cond2}\begin{split} \forall i\leq m,~i\neq j~ Q_{\lambda}(1\dots 1d_i1\dots 1d_j1\dots 1|\hat{u}_1\dots\hat{u}_{i-1}\hat{v}_i\hat{u}_{i+1}\dots\hat{u}_{j-1}\hat{v}_{j}\hat{u}_{j+1}\dots\hat{u}_m) &=0, \\ \forall ~i>m,~Q_{\lambda}(1\dots 1d_j1\dots 1|\hat{u}_1\dots \hat{u}_{j-1}\hat{v}_{j}\hat{u}_{j+1}\dots \hat{u}_m)
R_{\lambda}(1\dots 1d_{i}1\dots 1|\hat{u}_{m+1}\dots \hat{u}_{i-1}\hat{v}_{i}\hat{u}_{i+1}\dots \hat{u}_{N}) &=0.\end{split}\een
The last equation contradicts the condition of (\ref{cond1}). Also, for $j>m$
\ben\label{cond2}\begin{split} \forall i\leq m, Q_{\lambda}(1\dots 1 d_i 1\dots 1|\hat{u}_1\dots \hat{u}_{i-1}\hat{v}_i\hat{u}_{i+1}\dots \hat{u}_m)
R_{\lambda}(1\dots 1d_{j}1\dots 1|\hat{u}_{m+1}\dots \hat{u}_{j-1}\hat{v}_{j}\hat{u}_{j+1}\dots \hat{u}_{N}) &=0.\end{split}\een
\end{widetext}
But this again contradicts the condition of (\ref{cond1}). Hence, the assumed term can not be bi-local in any cut. The proof for $j> m$ is quite similar.
\end{proof}
In the next section we would like to interpret the strength of the argument(\ref{relH}) with respect to the conventional extension of Hardy's argument\cite{Chudhary2010} in a minimally constrained theory i.e in generalized no-signaling theory. To do so we need to compare the optimal paradoxical probabilities of both the arguments, as in \cite{Chudhary2010} and the present work.
\subsection{Multi-partite Hardy Paradox in GNST}
\label{sec_nst}
Here we study the more natural extension of Hardy's paradox for higher dimensional systems within the framework of generalized probabilistic theories. The only condition that we impose on the generalized probability distribution is the no-signaling condition, which all known physical theories respect. Normalization condition has also been imposed as a natural restriction on the probability measure.
\par
\emph{(1)Normalization conditions:} The probability distribution relating the outcomes for a given measurement setting should satisfy the normalization condition.

\begin{equation}\label{normal}
\sum_{x_1=1}^{d_1}\dots \sum_{x_N=1}^{d_N}P(x_1\dots x_N|X_1\dots X_N)=1 \mbox{ $\forall X_i\in \{\hat{u}_i,\hat{v}_i\}$.}
\end{equation}

where $i\in\{1,...,N\}$.
\par
\emph{(2)Non-signaling condition:} For no-signaling n-partite distribution $P(x_1x_2x_3\dots |X_1X_2X_3\dots )$ holds the fact that, each subset of parties $\{x_1,x_2,\dots ,x_m\}$ only depends on its corresponding inputs, i.e. if we change the input of one party it does not effect the marginal probability distribution for the other spatially separated parties.
\par
For a multi-partite generalized probability distribution the no-signaling conditions take the following form
\begin{eqnarray}\label{ns}
&\sum_{x_1=1}^{d_1} P(x_1x_2\dots x_N|X_1X_2\dots X_N)&\nonumber\\&=\sum_{x'_1=1}^{d_1}P(x'_1x_2\dots x_N|X'_1X_2\dots X_N)& \\
&\forall X_1,X_2,...,X_N,x_1,x_2,..,x_N,X'_1,x'_1.&\nonumber
\end{eqnarray}
This condition also holds for the cyclic permutation of the parties.
\subsubsection*{Results for multipatite system in GNST}
At this point we study the optimal success probability of new Hardy-type test (Eq.(\ref{relH})) under generalized no-signaling theory. This is to understand the structure of the modified multi-party paradox. Since under GNST all the constraints are linear one can easily see that finding out the optimal probability of success of the set of arguments in Eq.(\ref{relH}) is a linear programming problem. We have studied the case of $3$ and $4$ parties with dimension ranging from $2$ to $5$.
\begin{observation}\label{obs}
The optimal value of the probability of success, $q$ in(\ref{relH}) is $\frac{1}{3}$ for $n=3,4$ and $d=2,3,4,5$ in GNST.
\end{observation}
To find the optimal value of $q$ given in (\ref{relH}) 
one needs to optimize the probability of occurrence of the event $(11\dots 1|\hat{u}_1\hat{u}_2\dots \hat{u}_N)$ subject to the linear constraints (\ref{normal}), (\ref{ns}) and (\ref{relH}). The optimal success probability in(\ref{relH}) under GNST is in sharp contrast with the result of conventional extension of Hardy's paradoxical probability which is $\frac{1}{2}$. The observation.\ref{obs} leads us to the following conjecture
\begin{conjecture}
The optimal value of the probability of success, $q$ in(\ref{relH}) is $\frac{1}{3}$ for any arbitrary dimension for any number of parties in GNST.
\end{conjecture}

\section{Conclusion}
Multi-party version of Hardy's argument is the simplest to demonstrate non-locality. In a stronger form \cite{Rahaman2014,Chen2014} it has also been shown to be useful to detect {\it genuineness} of multipartite nonlocal correlations arising from qubit or very restricted systems. In this work we have generalized the argument such a way that {\it genuineness} can be detected in multipartite nonlocal correlations even for higher dimensional systems. We also show that this generalized argument is stronger than the conventional extension of Hardy's argument in the sense that the paradoxical probability for (\ref{relH}) is significantly less compared to the argument in \cite{Chudhary2010} under a minimally constrained theory, namely the generalized no-signaling theory(GNST).

\textbf{Acknowledgment}: Authors thank Guruprasad Kar, Manik Banik, Samir Kukri for useful and valuable discussion. RR acknowledges fruitful discussions with Marcin Wie\'{s}niak and Marek \.{Z}ukowski. AM acknowledge support from the CSIR project 09/093(0148)/2012-EMR-I and RR acknowledges support from UGC (University Grants Commission, Govt. Of India) Start-Up Grant.

\end{document}